\documentclass[twoside,journal]{IEEEtran}
\normalsize

\usepackage{bm}
\usepackage{array}
\usepackage{mathrsfs}
\usepackage{color}
\usepackage{cite}
\usepackage{graphicx}
\usepackage{latexsym}
\usepackage{amsmath,amsthm,amssymb, bbm}
\usepackage{stfloats}
\usepackage{tabularx}
\usepackage{enumerate}
\usepackage{algorithm}
\usepackage{algorithmic}
\usepackage{multirow}
\usepackage{paralist}
\usepackage{cases}
\usepackage[caption=false,font=footnotesize]{subfig}

 \newcommand{\be}{\begin{equation}}
 \newcommand{\ee}{\end{equation}}
 \newcommand{\bee}{\begin{eqnarray}}
 \newcommand{\eee}{\end{eqnarray}}
 
 \newtheorem{theo}{Theorem} 
 \newtheorem{lem}{Lemma} 
 
\newcommand{\algorithmiclastcon}{\textbf{Output:}}
\newcommand{\algorithmicfirstcon}{\textbf{Input:}}
\newcommand{\firstcon}{\item[\algorithmicfirstcon]}
\newcommand{\lastcon}{\item[\algorithmiclastcon]}

\begin{document}

\title{Wireless Information and Power Transfer Design for Energy Cooperation Distributed Antenna Systems}

\author{
        Fangchao~Yuan,
        Shi~Jin,~\IEEEmembership{Member,~IEEE,}
        Kai-Kit Wong,~\IEEEmembership{Fellow,~IEEE,}
        and~Hongbo~Zhu

\thanks{Fangchao~Yuan, and Hongbo~Zhu are with the Jiangsu Key Laboratory of Wireless Communications, College of Telecommunications and Information Engineering, Nanjing University of Posts and Telecommunications, Nanjing 210003, China. (e-mail: dylanyfc@163.com, zhb@njupt.edu.cn).}
        \thanks{Shi~Jin is with the National Mobile Communications Research Laboratory, Southeast University, Nanjing 210096, China. (e-mail: jinshi@seu.edu.cn). }
        \thanks{Kai-Kit Wong is with the Department of Electronic and Electrical Engineering, University College London, London WC1E 7JE, U. K. (e-mail: kai-kit.wong@ucl.ac.uk).}}

\maketitle

\begin{abstract}
Distributed antenna systems (DAS) have been widely implemented in state-of-the-art cellular communication systems to cover dead spots. Recent studies have also indicated that DAS have advantages in wireless energy transfer (WET). In this paper, we study simultaneous wireless information and power transfer (SWIPT) for a multiple-input single-output (MISO) DAS in the downlink which consists of arbitrarily distributed remote antenna units (RAUs). In order to save the energy cost, we adopt energy cooperation of energy harvesting (EH) and two-way energy flows to let the RAUs trade their  harvested energy through the smart grid network. Under individual EH constraints, per-RAU power constraints and various smart grid considerations, we investigate a power management strategy that determines how to utilize the stochastically spatially distributed harvested energy at the RAUs and how to trade the energy with the smart grid simultaneously to supply maximum wireless information transfer (WIT) with a minimum WET constraint for a receiver adopting power splitting (PS). Our analysis shows that the optimal design can be achieved in two steps. The first step is to maximize a new objective that can simultaneously maximize both WET and WIT, considering both the smart grid profitable and smart grid neutral cases. For the grid-profitable case, we derive the optimal full power strategy and provide a closed-form result to see under what condition this strategy is used. On the other hand, for the grid-neutral case, we illustrate that the optimal power policy has a double-threshold structure and present an optimal allocation strategy. The second step is then to solve the whole problem by obtaining the splitting power ratio based on the minimum WET constraint. Simulation results are provided to evaluate the performance under various settings and characterize the double-threshold structure.
\end{abstract}

\begin{IEEEkeywords}
Energy harvesting, distributed antennas, simultaneous wireless information and power transfer, smart grid.
\end{IEEEkeywords}

\IEEEpeerreviewmaketitle

\section{Introduction}
\IEEEPARstart{E}{nergy harvesting} (EH) traditionally refers to the extraction of energy from ambient environment for cost-effective and self-sustainable operation \cite{BIB01,BIB02,BIB03,Yuan}. However, energy harvested from ambient environments is passive, unreliable and uncontrollable to yield useful energy when needed. A new trend hence has emerged to use radio-frequency (RF) purposefully to transfer energy over the air to charge devices for their communications. This technology enables proactive energy replenishment of wireless devices, resulting in advantages in supporting applications with quality-of-service (QoS) requirement. True mobility would be achieved because mobile devices no longer depend on centralized power sources.

Wireless energy transfer (WET) has long been considered as a possibility, dating back to as early as 1891 in Tesla's demonstration~\cite{Tesla}. On the other hand, radio signals have since been widely used for wireless information transmission (WIT). As a consequence, it is reasonable that simultaneous wireless information and energy transfer (SWIPT) has recently drawn an upsurge of interests, see e.g., \cite{Varshney,Grover,ZR1,ZR2}. Using this technique, mobile users are provided with not only wireless data but also access to reliable energy supply at the same time.

Varshney was the first to propose the idea of SWIPT which was published in \cite{Varshney}, where he characterized the fundamental performance tradeoff with a capacity-energy function. Later, \cite{Grover} extended the result to frequency-selective channels. More recently, optimal design of different outage for the energy/rate tradeoffs was studied in~\cite{ZR1} subject to co-channel interference. In \cite{ZR2}, practical receiver designs were investigated for SWIPT. One major concern for SWIPT is its drastically decaying WET efficiency over the distance due to propagation loss. To alleviate this, multiple-input multiple-output (MIMO) beamforming systems \cite{ZR3,ZR4,ZR5,beamforming1,beamforming2,beamforming3,beamforming4,beamforming5}  have been proposed to help improve the WET efficiency. In particular, the authors in~\cite{ZR3} characterized the various achievable rate-energy (R-E) tradeoffs by practical receiver designs. The results have subsequently been extended to massive MIMO~\cite{ZR4} and multiuser channel setups~\cite{ZR5,beamforming1}.

In order to achieve efficient WET, electromagnetic energy should be concentrated into a sharp narrow beam, referred to as energy beamforming which was studied in~\cite{ZR3,beamforming1,beamforming2,beamforming3,beamforming4,beamforming5} for different scenarios. Nevertheless, only those users close to the energy transmitter can harvest meaningful energy, while those far away from the transmitter will get much less power. Such distance limitation for WET can be prevented if energy receivers are brought closer to the transmitters in distributed antenna systems (DASs) \cite{Yuan2}. Specifically, the remote antenna units (RAUs) are more arbitrarily distributed over the cells, and the distance from any given user equipment (UE) to its nearby RAU(s) is much smaller, making SWIPT more viable.

While there is strong interest to use EH to reduce or even replace the energy purchased from the grid, the harvested energy highly depends on  environmental factors such as location and weather, and is random and intermittent by nature. Hence, it is difficult to maintain the prescribed QoS given the uncertainty of the available power for each RAU if it relies solely on EH. A better approach is therefore to have both EH and the grid to power the RAUs \cite{grid1,grid2,grid3,grid4,grid5}. Recent advances in smart grids further enable power trading amongst consumers via the use of smart meters \cite{sgrid1,sgrid2,sgrid3,sgrid4}. Two-way energy flows are possible between the grid and the RAUs, facilitating also the RAUs to trade their unevenly harvested energy through the smart grid. The fact that the mismatch between EH and the RAU's power demand leads to energy outage or energy wastage (insufficient or excessive harvested energy) needs to be addressed.

Overall, we see three critical challenges: SWIPT, EH wireless systems, and smart grid enabled EH wireless systems. To the best of our knowledge, prior works tended to concentrate upon one or two of the three challenges. Motivated by this, in this paper, we consider a DAS for SWIPT with EH capability and smart grid coexisting, which involves addressing all three challenges jointly. In particular, new challenges arise for the design of power management for each RAU with random EH and the corresponding trade management with smart grid when serving the space-dependent and time-varying SWIPT traffic. The fact that the harvested energy is typically much cheaper than the energy purchased from the grid, also motivates the maximization of the use of the harvested energy to save cost. It is therefore reasonable to assume that the smart grid will not be ``trade-deficit" during all trading with the RAUs.


In this paper, we focus on the use of power splitting (PS) receivers permitting each user to receive both information and energy from the RAU continuously at all time. Time switching (TS) receivers can be considered as a special case of PS with only binary split power ratios \cite{ZR2,ZR3}, and therefore PS can in general achieve better rate-energy transmission trade-off than TS. We show that the optimal design can be achieved in two steps. The first is to maximize the WET and WIT performance for two cases, namely the smart grid profitable and smart grid neutral cases, with the full-power and double-threshold power allocation policies, respectively. Then the problem is addressed by finding the PS ratio from the minimum WET constraint.


The rest of the paper is organized as follows. The system model is presented in Section II. Section III analyzes the characteristics of the optimal power allocation policy and also provides some lemmas. We address the smart grid profitable case in Section IV, while the smart grid neutral case will be tackled in Section V. Numerical results are presented and compared in Section VI. Section VII concludes the paper.


\section{System Model}
\begin{figure*}[!t]
\centering
\includegraphics[scale=0.8]{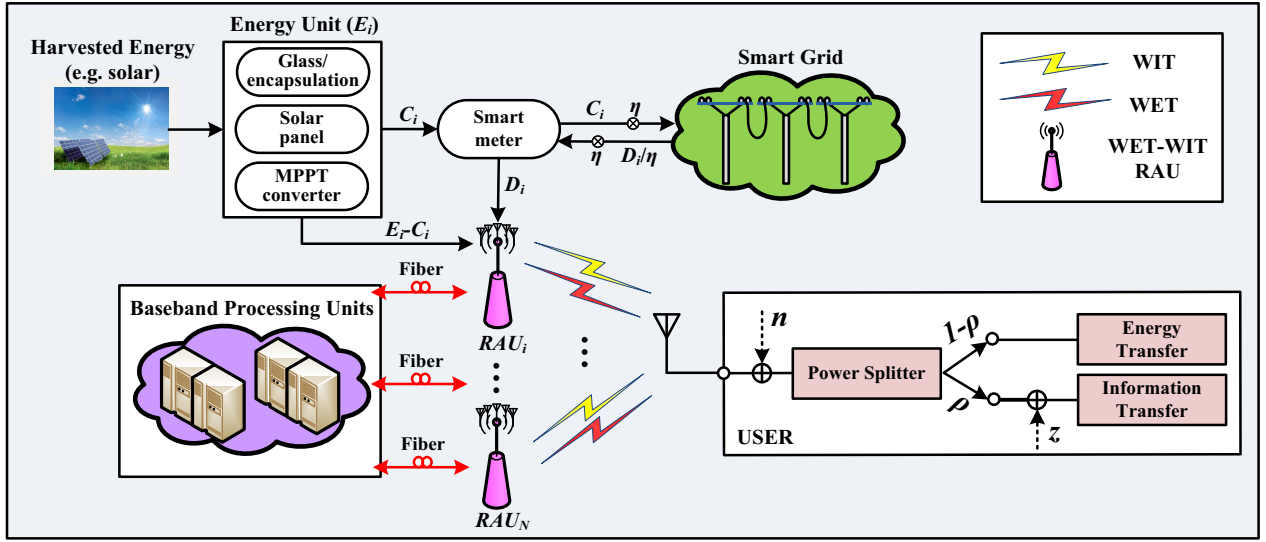}
\caption{The system model where smart grid techniques enable power trading among the consumers via smart meters permitting power transmission between the RAUs and the power grid.}\label{Fig:2}
\end{figure*}

In this section, we introduce our model for energy cooperation and SWIPT for cellular systems, which is depicted in Fig.~\ref{Fig:2}. We consider a downlink single-cell DAS that has $N$ RAUs with $M>1$ antennas each, all connected to a baseband processing unit (BPU) using high-quality bidirectional wired (e.g., radio-over-fiber) or wireless (e.g., microwave repeater) links. The BPU is assumed to have all the necessary baseband processing capability of a base station (BS), and the harvested energy profile of all the RAUs. Moreover, it is assumed that the DAS is to serve only one user for simplicity, and all RAUs and the user know channel state information perfectly. Note that all the power allocation is centrally controlled. Then the received signal for the user is written as
\begin{equation}
{y} = \sum\limits_{i = 1}^N {\sqrt {{p_i}} } \mathbf{g}_{i}^H{\mathbf{w}_{i}}{s} + {n},
\end{equation}
where $p_i$ is the transmit power consumed by the $i$-th RAU, ${\mathbf{w}_{i}}$ is defined as the beamforming column vector of length $M$ for the $i$-th RAU with unit norm (${{\| {{\mathbf{w}_{i}}}\|}^2} = 1$), $s$ stands for the transmitted signal with zero mean and unit variance, and $n$ denotes the zero-mean additive white Gaussian noise (AWGN) with variance ${\sigma ^2}$, the superscript $(\cdot)^H$ is the Hermitian operation, and $\mathbf{g}_{i}$ is the $M \times1$ channel vector to characterize the channel state, which is modelled by
\begin{equation}
\mathbf{g}_{i}= \mathbf{h}_{i} \sqrt {{\beta _{i}}},
\end{equation}
where $\mathbf{h}_{i}$ indicates the channel column vector for small-scale fading with each entry assumed to have zero mean and unity variance, and $\beta _{i}$ accounts for the large-scale fading which can be factored into
\begin{equation}
{\beta _{i}} = \frac{{{1}}}{{d_{i}^\alpha }},
\end{equation}
where $\alpha$ is the decay exponent. The received signal power is split for simultaneous WIT and WET, a $\rho  \in \left( {0,1} \right]$ portion of the received signal power to the WIT, and the remaining $\left({1-\rho}\right)$ portion for the WET under SWIPT. As a result, the split signal for the WET of the user is given as
\begin{equation}
y^{\rm WET} = \sqrt {{1-\rho}} y.
\end{equation}
Accordingly, the energy transferred to the user is proportional to the split signal which is given by
\begin{equation}\label{EH}
  {Q} = \xi \left({1-\rho}\right) \left({{\left| {\sum\limits_{i = 1}^N {\sqrt {{p_i}} } \mathbf{g}_{i}^H{\mathbf{w}_{i}}} \right|^2}+\sigma^2}\right),
\end{equation}
where $\xi  \in \left( {0,1} \right]$ stands for the energy conversion efficiency.

Furthermore, the split signal for WIT is expressed as
\begin{equation}
y^{\rm WIT} = \sqrt {{\rho }} y + {z},
\end{equation}
where $z$ is the AWGN with zero mean and variance ${\tau^2}$ during the WIT process. Thus, the user's achievable rate is given as
\begin{equation}\label{ID}
{R} = {\log _2}\left( {1 + \frac{ \rho {\left| {\sum\limits_{i = 1}^N {\sqrt {{p_i}} } \mathbf{g}_{i}^H{\mathbf{w}_{i}}} \right|^2} }{{{\rho}{\sigma ^2} + {\tau ^2}}}} \right).
\end{equation}

Each RAU $i$ is equipped with an EH device that delivers a harvesting energy rate $E_{i}>0$ at the beginning of transmission, for $i=1, \dots, N$. In practice, the value of $p_i$ varies according to the harvested energy at each RAU as well as the channel state information between the $i$-th RAU and the user. By combining the harvested energy rate $E_{i}$ and the transfer power $p_i$, we use $D_i$ to indicate the energy shortage leading to energy borrowing from the smart grid, i.e., {\em grid discharging}, and $C_i$ to present the energy surplus of renewable energy which can be paid/traded back to the smart grid, i.e., {\em grid charging}.  Since these two values are usually independent, it is likely that some RAUs are short of renewable energy to match demand (i.e., $D_i>0$), while the other RAUs are adequate in renewable energy (i.e., $C_i>0$). Such a geographical diversity requires some RAUs to borrow $D_i$ energy from the grid but the other RAUs to pay/trade back the extra renewable energy $C_i$ in order to trade or reuse the renewable energy by other RAUs rather than being wasted. These can be accommplished by energy cooperation which enables two-way energy flows between the smart grid and the RAUs~\cite{sgrid1,sgrid2,sgrid3,sgrid4}. From these two variables, the transmit power $p_i$ for each RAU is calculated as
\begin{equation}\label{Power}
p_i=E_i+D_i-C_i.
\end{equation}
Note that by definition, $C_i \ge 0$, $D_i \ge 0$, $p_i \ge 0~\forall i$. From (\ref{Power}), it imposes that $E_i+D_i-C_i \ge 0$ on $D_i$ and $C_i$.

It is impossible to deliver power between the RAUs and the smart grid without any loss, i.e., perfect sharing of power is not possible among the RAUs. There will be energy loss efficiency $\eta$ in electric power transmission between RAUs and the smart grid. To minimize the energy use from the grid, we assume that the smart grid will not be ``trade-deficit" during all trading with the RAUs, which we refer to as the {\em green smart constraint}. Specifically, this is formulated as
\begin{equation}
\sum\limits_{i = 1}^N { {{S _{i}}} }  \ge 0,
\end{equation}
where $S_i$ denotes the current trade state between the $i$-th RAU and the smart grid, which can be further described as
\begin{equation}\label{state}
{S_i} = \eta {C_i} - \frac{D_i}{\eta}.
\end{equation}

The coupling of SWIPT and the power utilization optimization introduces new challenges on the design of green energy-enabled wireless networks. In this paper, we aim to maximize the WIT performance with a minimum WET constraint and the per-RAU power constraint as well as the green smart constraint. From \eqref{EH} and \eqref{ID}, we can assert that when the WIT performance is maximized with $\rho$ remaining fixed, the WET performance is also maximum. Therefore, our optimization problem is first to maximize ${| {\sum\nolimits_{i = 1}^N {\sqrt {{p_i}} } \mathbf{g}_{i}^H{\mathbf{w}_{i}}} |^2}$ with the variable $\rho$ being fixed and then to adjust the value of the PS variable $\rho$ according to the WET constraint. Hence, we can see that the main problem is to tackle the former one. The latter can be easily solved by figuring out the splitting power ratio based on the minimum WET constraint using simple calculation. As a result, our focus is on solving the first problem, which is formulated as
\begin{subequations}\label{objective1}
\begin{align}
  \mathop {\max }\limits_{{p_i}, {\mathbf{w}_{i}}} ~~& {\left| {\sum\limits_{i = 1}^N {\sqrt {{p_i}} } \mathbf{g}_{i}^H{\mathbf{w}_{i}}} \right|^2}  \\
{\rm s.t.} ~~& \sum\limits_{i = 1}^N {{{S _{i}}} }  \ge 0\\
                                      &  0 \le p_i \le p_{max}, &\forall i \in \left\{ {1, 2,  \ldots, N} \right\} \\
                                      &  D_i \ge 0, ~C_i \ge 0, &\forall i \in \left\{ {1, 2,  \ldots, N} \right\} \\
                                      & {{\left\| {{\mathbf{w}_{i}}} \right\|}^2} = 1, &\forall i \in \left\{ {1, 2,  \ldots, N} \right\}
\end{align}
\end{subequations}
where $p_{max}$ is the per-RAU power constraint. In general, it is hard to obtain an explicit solution for this joint optimization problem since it is strictly non-convex. We first provide some useful lemmas for subsequent use. Then we characterize the optimal policy by wisely combining the methods of Lagrange multipliers and considering the problem in two scenarios.

\section{Analysis of the Optimal Policy}
We first provide an efficient beamforming strategy which achieves nearly the same performance compared to the optimal approach. From \cite{Lee}, we have the following lemma to design the beamforming vectors ${\mathbf{w}_{i}}$.

\begin{lem}
For given $p_i$'s, the optimal beamforming solution is identical to that of \cite{Lee} under a single-user multiple-input single-output (MISO) DAS, i.e., the distributed maximum ratio transmission (DMRT), or ${\mathbf{w}_{i}^*}=\frac{\mathbf{g}_{i}}{{\left\| {{\mathbf{g}_{i}}} \right\|}}$ for $i=1, \ldots, N$.
\end{lem}
Different from the conventional iterative algorithms which optimize the two parameters in an alternating fashion, in our case, we set ${\mathbf{w}_{i}^*}=\frac{\mathbf{g}_{i}}{{\left\| {{\mathbf{g}_{i}}} \right\|}}$ independent of $p_i$. Thus, we do not need to update ${\mathbf{w}_{i}}$ iteratively.  Substituting ${\mathbf{w}_{i}^*}=\frac{\mathbf{g}_{i}}{{\left\| {{\mathbf{g}_{i}}} \right\|}}$ into the optimization, our problem can be reformulated as
\begin{subequations}\label{objective2}
\begin{align}
\mathop {\max }\limits_{{C_i},{D_i}} ~~& {\left( {\sum\limits_{i = 1}^N {\sqrt {p_i}{\gamma_i}}} \right)^2}  \label{objective2_1} \\
{\rm s.t.} ~~& \sum\limits_{i = 1}^N {{{S _{i}}} }  \ge 0\\
                                      &  0 \le p_i \le p_{max}, &\forall i \in \left\{ {1, 2,  \ldots, N} \right\} \\
                                      &  D_i \ge 0, ~C_i \ge 0, &\forall i \in \left\{ {1, 2,  \ldots, N} \right\}
\end{align}
\end{subequations}
where $\gamma_i$ is the update channel gain, equalling to ${d_{i}}^{-\frac{\alpha}{2}}{\left\| {{\mathbf{h}_{i}}} \right\|}$. Without loss of generality, we assume that all $\gamma_i$'s are sorted in descending order as ${\gamma_1} > {\gamma_2} > \cdots > {\gamma_N}$. In order to further derivations, we also introduce the following lemma.

\begin{lem}
For $\eta<1$, the solution to \eqref{objective2} satisfies $C_iD_i=0$ for all $i$, i.e., the optimal policy of each RAU never enables the grid charging and discharging simultaneously.
\end{lem}

\begin{proof}
Let $\left\{ {\left[ {{{C}_i},{{D}_i}} \right]} \right\}_{i = 1}^N$ be a feasible power policy which satisfies ${{C}_j}{{D}_j} > 0$ for some $j$. Let
\begin{align}
{{\bar C}_i} &= {\left[ {{C_i} - \frac{D_i}{\eta ^2}} \right]^ + },\\
{{\bar D}_i} &= {\left[ {{D_i} - {\eta ^2}{C_i}} \right]^ + },
\end{align}
where
\begin{equation}
{\left[ x \right]^ + } = \left\{ \begin{array}{ll}
x, & x \ge 0\\
0,   & x < 0. \end{array} \right.
\end{equation}
For all $i  \ne  j$, let ${{\bar C}_i}=C_i$ and ${{\bar D}_i}=D_i$. Note that the current trade states in \eqref{state} are unaffected by this change, since $\eta {C_i} - {D_i}/{\eta}=\eta {{\bar C}_i} - {{\bar D}_i}/{\eta}$, for all $i$. Therefore, the allocation policy $\left\{ {\left[ {{{\bar C}_i},{{\bar D}_i}} \right]} \right\}_{i = 1}^N$ is feasible. On the other hand, the resulting transmit power ${\bar p}_j$ at the $j$-th RAU becomes
\begin{equation}
{{\bar p}_j} = {E_j} + {{\bar D}_j} - {{\bar C}_j} = \left\{ \begin{array}{ll}
{E_j} + \frac{D_i}{\eta ^2} - {C_i},  & \eta {C_i}>\frac{D_i}{\eta},\\
{E_j} + {D_i} - {\eta ^2}{C_i},  & \mbox{otherwise},
\end{array} \right.
\end{equation}
and as such ${{\bar p}_j}>p_j$ due to $\eta<1$, and $C_i,D_i>0$. Since the objective function \eqref{objective2_1} is increasing in $p$, the power allocation policy $\left\{ {[ {{{\bar C}_i},{{\bar D}_i}} ]} \right\}_{i = 1}^N$ achieves better SWIPT than $\left\{ {[ {{{C}_i},{{D}_i}} ]} \right\}_{i = 1}^N$, and the latter policy cannot be optimal.
\end{proof}

We observe from \emph{Lemma 2} that we have $C_i > 0$ and $D_i=0$ (grid-charging), or $C_i = 0$ and $D_i>0$ (grid-discharging), or $C_i = 0$ and $D_i=0$ (referred to as grid-passive). The optimal policy does not store and retrieve energy simultaneously at any time. Through \emph{Lemma 2}, we can also see that
\begin{align}
{{C}_i} &= {\left[ {E_i-p_i} \right]^ + },\\
{{D}_i} &= {\left[ {p_i-E_i} \right]^ + }.
\end{align}
Since the problem now in \eqref{objective2} is a convex optimization problem,  Karush-Kuhn-Tucker (KKT) conditions are necessary and sufficient for optimality. Therefore, the Lagrangian function can be obtained as
\begin{multline}\label{general}
{\cal L}={\left( {\sum\limits_{i = 1}^N {\sqrt {p_i}{\gamma_i}}} \right)^2}+ \mathop \sum \limits_{i = 1}^N {\lambda _i}{p_i}+ \mathop \sum \limits_{i = 1}^N {\nu _i}{\left({p_{max}-p_i}\right)}\\
+ \mathop \sum \limits_{i = 1}^N {\theta _i}{C_i} + \mathop \sum \limits_{i = 1}^N {\xi _i}{D_i}+\mu{\mathop \sum \limits_{i = 1}^N S_i},
\end{multline}
where $\lambda_i$, $\nu_i$, ${\theta _i}$, ${\xi _i}$, $\forall i$ and $\mu$ are the non-negative Lagrange multipliers corresponding to the constraints.

The corresponding additional complimentary slackness conditions are given by
\begin{subequations}\label{slackness}
\begin{align}
{\lambda _i}{p_i}  =0, {\nu _i}{\left({p_{max}-p_i}\right)}&=0, \forall i \in \left\{ {1, \ldots N} \right\} \label{slackness_1}\\
{\theta _i}{C_i}  = 0, {\xi _i}{D_i} &=0, \forall i \in \left\{ {1, \ldots N} \right\} \label{slackness_2}\\
\mu{\mathop \sum \limits_{i = 1}^N S_i}& = 0. \label{slackness_3}
\end{align}
\end{subequations}
Next, we derive the power allocation solution for DAS with DMRT.  Based on the additional complimentary slackness condition~\eqref{slackness_3}, we design the strategy by separately solving two scenarios, namely, {\em smart grid profitable:} ${\mathop \sum \nolimits_{i = 1}^N S_i}>0$ and {\em smart grid neutral:} ${\mathop \sum \nolimits_{i = 1}^N S_i}=0$ and then give the entire optimal policy based on the strategy of each scenario.

\section{Optimal Design of Power Allocation for SCENARIO I: Smart Grid Profitable}
In this scenario, we can know that ${\mathop \sum \nolimits_{i = 1}^N S_i}>0$, and can obtain $\mu=0$ from (\ref{slackness_3}). The KKT optimality conditions are found by taking the derivatives with respect to $C_k$ and $D_k$ for
$k = \left\{ {1, \ldots N} \right\}$ as
\begin{align}\label{Lar_pro}
\frac{{\partial {\cal L}}}{{\partial {C_k}}} &= -\frac{{{\gamma _k}\sum\nolimits_{i = 1}^N {\sqrt {{p_i}} {\gamma _i}} }}{{\sqrt {{p_k}} }} - {\lambda _k} + {\nu _k}+{\theta_k},\\
\frac{{\partial {\cal L}}}{{\partial {D_k}}} &= \frac{{{\gamma _k}\sum\nolimits_{i = 1}^N {\sqrt {{p_i}} {\gamma _i}} }}{{\sqrt {{p_k}} }} + {\lambda _k} - {\nu _k}+{\xi_k}.
\end{align}
The optimal power allocation $p_k$ of the $k$-th RAU can be divided into three mutually exclusive cases according to the additional complimentary slackness conditions (\ref{slackness_1}):
\begin{align}\label{Three_cases}
\left( {p_k^ * ,\lambda _k^ * ,\nu _k^ * } \right) = \left\{ \begin{array}{l}
\left( {0,\lambda _k^ * ,0} \right),\\ [12pt]
\left( {p_k^ * ,0,0} \right)\left| {0 < p_k^ *  < {p_{\max }},} \right.\\[12pt]
\left( {{p_{\max }},0,\nu _k^ * } \right)
\end{array} \right.
\end{align}

\begin{theo}
The full-power policy: if the smart grid is to be profitable, i.e., ${\mathop \sum \nolimits_{i = 1}^N S_i}>0$, then the optimal strategy is that all the RAUs transmit with the maximum power $p_{max}$, i.e., $p_i^*=p_{max}$, $\forall i \in \{1,\ldots,N\}$.
\end{theo}

\begin{proof}
We consider the three mutually exclusive cases.
\subsection*{Case 1: $\left( {0,\lambda _k^ * ,0} \right)$}
When $p_k^ *=0$, we know $\nu_k=0$. Thus, setting the Lagrange function to zero leads to
\begin{align}
\frac{{\partial {\cal L}}}{{\partial {C_k}}} &= -\frac{{{\gamma _k}\sum\nolimits_{i = 1}^N {\sqrt {{p_i}} {\gamma _i}} }}{{\sqrt {{p_k}} }} - {\lambda _k} +{\theta_k}=0,\label{Lar_prof1_1}\\
\frac{{\partial {\cal L}}}{{\partial {D_k}}} &= \frac{{{\gamma _k}\sum\nolimits_{i = 1}^N {\sqrt {{p_i}} {\gamma _i}} }}{{\sqrt {{p_k}} }} + {\lambda _k} +{\xi_k}=0.\label{Lar_prof1_2}
\end{align}
With the power $p_k^ *=0$, the left-hand-side of \eqref{Lar_prof1_2} becomes infinity, which contradicts the KKT conditions. Therefore, it is impossible to allocate zero power to the $k$-th RAU for $\forall k$.

\subsection*{Case 2: $\left( {p_k^ * ,0,0} \right)\left| {0 < p_k^ *  < {p_{\max }}} \right.$}
When the power falls within the constraint $0 < p_k^ *< {p_{\max }}$, we know $\lambda_k=0$ and $\nu_k=0$. Similarly, setting the Lagrange function to zero leads to
\begin{align}
\frac{{\partial {\cal L}}}{{\partial {C_k}}} &= -\frac{{{\gamma _k}\sum\nolimits_{i = 1}^N {\sqrt {{p_i}} {\gamma _i}} }}{{\sqrt {{p_k}} }} +{\theta_k}=0,\label{Lar_prof2_1}\\
\frac{{\partial {\cal L}}}{{\partial {D_k}}} &= \frac{{{\gamma _k}\sum\nolimits_{i = 1}^N {\sqrt {{p_i}} {\gamma _i}} }}{{\sqrt {{p_k}} }} +{\xi_k}=0.\label{Lar_prof2_2}
\end{align}
From Case 1, we know all the power $p_k^ *>0, \forall k$. Combining $\gamma_k>0, {\xi_k}\ge 0, {\theta_k}\ge 0, \forall k$, the left-hand-side of the above equation is greater than zero, which contradicts with the KKT conditions. Thus, the optimal transmission power of the $k$-th RAU is $p_{max}$ for $\forall k \in \{1,\dots,N\}$.
\end{proof}

This theorem illustrates that if the harvested energy profile is good enough, not only will the user receive the maximum SWIPT, but also the smart grid will be power replenished. This theorem also reveals an interesting strategy in the optimal power allocation pattern, as summarized below.

\begin{lem}
Assuming $p_i=p_{max}$, $\forall i \in {1,\ldots,N}$,  if we can obtain that
\begin{align}\label{criterion}
\sum\limits_{i = 1}^N {{E_i}}  \ge N{p_{\max }} + \left( {1 - {\eta ^2}} \right){\sum\limits_{i \in G} {\left( {{E_i} - {p_{\max }}} \right)}},
\end{align}
where $G$ denotes the set of the RAUs whose EH rate is greater than the transmit power (i.e., grid-charging), then $p_i^*=p_{max}$, $\forall i \in {1,\ldots,N}$ is the optimal power allocation strategy.
\end{lem}

\begin{proof}
Since $\left\{ G \right\} + \left\{ L \right\} + \left\{ {PA} \right\} = \left\{ {1,2, \ldots ,N} \right\}$, where  the set $L$ denotes the RAUs whose harvesting rate is less than the transmit power (i.e., grid-discharging) and $\left\{ {PA} \right\}$ is the set for the RAUs whose harvesting rate is equal to the transmit power (i.e., grid-passive),  we rearrange \eqref{criterion} as
\begin{multline}
\underbrace {\sum\limits_{i = 1}^N {{E_i}}  - \sum\limits_{i \in G} {{E_i}} }_{\sum\limits_{i \in L} {{E_i}}+\sum\limits_{i \in PA} {{E_i}} } + {\eta ^2}\sum\limits_{i \in G} {{E_i}}  \ge \\
\underbrace {N{p_{\max }} - \sum\limits_{i \in G} {{p_{\max }}} }_{\sum\limits_{i \in L} {{p_{\max }}}+\sum\limits_{i \in PA} {{p_{\max }}}  } + {\eta ^2}\sum\limits_{i \in G} {{p_{\max}}}.
\end{multline}
Since $\sum\limits_{i \in PA} {{E_i}}=\sum\limits_{i \in PA} {{p_i}}=\sum\limits_{i \in PA} {{p_{\max }}}$, we have
\begin{align}
{\eta ^2}\sum\limits_{i \in G} {\left( {{E_i} - {p_{\max }}} \right)}  & \ge \sum\limits_{i \in L} {\left( {{p_{\max }} - {E_i}} \right)},\\
\sum\limits_{i \in G} {\eta \left( {{E_i} - {p_{\max }}} \right)} &  \ge \sum\limits_{i \in L} {\frac{{\left( {{p_{\max }} - {E_i}} \right)}}{\eta }}.\label{GL}
\end{align}
Note that grid-charging imposes that $C_i=E_i-p_i$ with $i\in G$. Similarly, grid-discharging imposes that $D_i=p_i-E_i$ with $i\in L$. Therefore, \eqref{GL} becomes
\begin{align}\label{CD}
\sum\limits_{i \in G} {\eta  {C_i} }  \ge \sum\limits_{i \in L} {\frac{{ {D_i} }}{\eta}}.
\end{align}
From \emph{Lemma 2}, it is known that the optimal power allocation policy for each RAU never enables grid charging and discharging simultaneously. Thus, we have $C_i=0, i \in \{L, PA\}$ and $D_i=0, i \in \{G, PA\}$. Now \eqref{CD}  becomes
\begin{align}
\sum\limits_{i \in N} {\eta  {C_i}}  & \ge \sum\limits_{i \in N} {\frac{{ {D_i}}}{\eta }},\\
\sum\limits_{i = 1}^N \left({{\eta  {C_i}} -{\frac{{ {D_i}}}{\eta }}}\right) & \ge 0,\\
\sum\limits_{i = 1}^N {{S_i}}  & \ge 0.
\end{align}
Now, the proof is replaced by that when $p_i=p_{max}$, $\forall i$,  if $\sum_{i = 1}^N {{S_i}}\ge 0$, $p_i^*=p_{max}$ is the optimal power strategy.

We prove this by {\em reductio ad absurdum}. Assume that when $p_i=p_{max}$, $\forall i$,  if $\sum\nolimits_{i = 1}^N {{S_i}}\ge 0$ and the optimal power strategy is not $p_i^*=p_{max}$, $\forall i \in {1,\ldots,N}$, that is to say, not all the RAUs transmit with the maximum power (at least one RAU transmits with the power that is lower than the maximum). Let us now say that $\exists j \in A$, where $A$ denotes the set of the RAUs whose transmission power is lower than the maximum power and is a non-empty set, i.e., $p_j<p_{max}$.

As when $p_i=p_{max}$, $\forall i \in {1,\ldots,N}$, $\sum\nolimits_{i = 1}^N {{S_i}}\ge 0$. From this, we know that when the optimal power allocation $\exists j \in A$ which enables $p_j^*<p_{max}$, the sum of the current trade state will be greater than zero ($\sum\nolimits_{i = 1}^N {{S_i}}> 0$). From \emph{Theorem 1}, we obtain that when the grid is profitable, the optimal strategy is that all the RAUs transmit with the maximum power $p_{max}$, which conflicts with the hypothetical proposition.
\end{proof}

This lemma indicates that the first step to find the optimal strategy is to test if the sum of the harvested energy $\sum\nolimits_{i = 1}^N {{E_i}}$ is greater than or equal to $N{p_{\max }} + \left( {1 - {\eta ^2}} \right) {\sum\nolimits_{i \in G} {\left( {{E_i} - {p_{\max }}} \right)} }$. If yes, the optimal power allocation is settled. Otherwise, the sum of the trade states will be less than zero, $\sum\nolimits_{i = 1}^N {{S_i}}<0$ with $p_i=p_{max}$, which conflicts with the slackness conditions \eqref{slackness_3}. Thus, we need to decrease the power of some RAUs to enable $\sum\nolimits_{i = 1}^N {{S_i}}=0$, which is solved in the next section.

\section{Optimal Design of Power Allocation for SCENARIO II: Smart Grid Neutral}
Based on the last section, we move towards our final policy. With ${\mathop \sum \nolimits_{i = 1}^N S_i}=0$, the KKT optimality conditions are found by taking the derivatives of \eqref{general} with respect to $C_k$ and $D_k$ for $k = \left\{ {1, \ldots N} \right\}$ as
\begin{align}\label{Lar_neu}
\frac{{\partial {\cal L}}}{{\partial {C_k}}} &= -\frac{{{\gamma _k}\sum\nolimits_{i = 1}^N {\sqrt {{p_i}} {\gamma _i}} }}{{\sqrt {{p_k}} }} - {\lambda _k} + {\nu _k}+{\theta_k}+{\mu}{\eta},\\
\frac{{\partial {\cal L}}}{{\partial {D_k}}} &= \frac{{{\gamma _k}\sum\nolimits_{i = 1}^N {\sqrt {{p_i}} {\gamma _i}} }}{{\sqrt {{p_k}} }} + {\lambda _k} - {\nu _k}+{\xi_k}- \frac{\mu}{\eta }.
\end{align}
By setting ${\partial {\cal L}} /{\partial {C_k} }={\partial {\cal L}} /{\partial {D_k} = 0}$, we obtain
\begin{subequations}
\begin{align}
\frac{{\sqrt {{p_k}} }}{\gamma_k} & = \frac{{\sum\nolimits_{i = 1}^N {\sqrt {{p_i}} {\gamma _i}} }}{{\mu}{\eta }- {\lambda _k} + {\nu _k}+{\theta_k}},\label{property_1}\\
&= \frac{{\sum\nolimits_{i = 1}^N {\sqrt {{p_i}} {\gamma _i}} }}{\frac{\mu}{\eta }- {\lambda _k} + {\nu _k}-{\xi_k}}.\label{property_2}
\end{align}
\end{subequations}
Similarly, the optimal allocation $p_k$ of the $k$-th RAU can be divided into three mutually exclusive cases according to the additional complimentary slackness conditions (\ref{slackness_1}):
\begin{align}\label{Three_cases_2}
\left( {p_k^ * ,\lambda _k^ * ,\nu _k^ *, \mu^*} \right) = \left\{ \begin{array}{l}
\left( {0,\lambda _k^ * ,0, \mu^*} \right),\\ [12pt]
\left( {p_k^ * ,0,0, \mu^*} \right)\left| {0 < p_k^ *  < {p_{\max }},} \right.\\[12pt]
\left( {{p_{\max }},0,\nu _k^ *, \mu^* } \right).
\end{array} \right.
\end{align}
First, we determine the properties of the optimal solution in the following three lemmas for the three cases.

\begin{lem}
For any $k$ and $j$, if the optimal power of the $k$-th grid-charging RAU is $p_{max}$, then the power for the $j$-th grid-charging RAU having better signal-to-noise ratio (SNR) than the $k$-th RAU is determined as $p_{max}$. If the optimal power of the $k$-th grid-discharging RAU is $p_{max}$, the power for the $j$-th RAU having better SNR than the $k$-th RAU is determined as $p_{max}$ regardless of the current trade state of the $j$-th RAU.
\end{lem}

\begin{proof}
When $p_k^*=p_{max}$ and $p_k^*<E_k$ ($C_k>0$), combining the slackness conditions \eqref{slackness_1} and \eqref{slackness_2} gives
\begin{equation}
\frac{{\sqrt {{p_k}} }}{\gamma_k}  = \frac{{\sum\nolimits_{i = 1}^N {\sqrt {{p_i}} {\gamma _i}} }}{{\mu}{\eta } + {\nu _k}}.
\end{equation}
As $p_j^*<E_j$, we obtain
\begin{equation}
\frac{{\sqrt {{p_j}} }}{\gamma_j}  = \frac{{\sum\nolimits_{i = 1}^N {\sqrt {{p_i}} {\gamma _i}} }}{{\mu}{\eta } -\lambda_j + {\nu _j}}.
\end{equation}
\begin{itemize}
\item If $p_j=0$, we know $\nu_j=0$ from \eqref{slackness_1}, leading to
\begin{equation}
\frac{{\sqrt {{p_j}} }}{\gamma_j}  = \frac{{\sum\nolimits_{i = 1}^N {\sqrt {{p_i}} {\gamma _i}} }}{{\mu}{\eta } -\lambda_j }>\frac{{\sqrt {{p_k}} }}{\gamma_k}.
\end{equation}
Since ${\gamma_j}>{\gamma_k}$, $p_j$ should be greater than $p_k$ ($p_j>p_k=p_{max}$) to ensure the above relationship, which conflicts the original assumption $p_j=0$. As a result, the power allocation policy $p_j=0$ cannot be optimal.
\item If $0<p_j<p_{max}$, $\lambda_j=\nu_j=0$, leading to
\begin{equation}
\frac{{\sqrt {{p_j}} }}{\gamma_j}  = \frac{{\sum\nolimits_{i = 1}^N {\sqrt {{p_i}} {\gamma _i}} }}{{\mu}{\eta } }>\frac{{\sqrt {{p_k}} }}{\gamma_k}.
\end{equation}
As ${\gamma_j}>{\gamma_k}$, $p_j$ should be greater than $p_k$ ($p_j>p_k=p_{max}$) to ensure the above relationship, which conflicts the original assumption $0<p_j<p_{max}$. Hence, again, this policy $0<p_j<p_{max}$ cannot be optimal.
\end{itemize}
In summary, the optimal policy for the $j$-th grid-charging RAU is $p_j^*=p_{max}$. For the grid-discharging RAUs, the proof is similar to the grid-charging case. No matter what the current trade state between the $j$-th RAU and the smart grid is, if $p_j \ne p_{max}$, $\frac{{\sqrt {{p_j}} }}{\gamma_j}$ is always greater than ${\gamma_j}>{\gamma_k}$, and thus $p_j$ should be greater than $p_k$ ($p_j>p_k=p_{max}$), which conflict the assumption $p_j \ne p_{max}$ and completes the proof.
\end{proof}

Importantly, the RAU allocated with the maximum power is for the better channel gain because the energy loss in trade with the smart grid encourages spending the harvested energy directly at the current RAU to ``trade for" the data rate and energy transfer performance to avoid energy ``devaluation".

\begin{lem}
For any $k$ and $j$, if the optimal power of the $k$-th RAU is zero, the power for the $j$-th RAU having worse SNR than the $k$-th RAU is determined as 0.
\end{lem}

\begin{proof}
When $p_k^*=0$, we know that this RAU must be a grid-charging RAU ($C_k>0$) as the harvesting rate of the $k$-th RAU is certainly greater than zero. Combining the slackness conditions \eqref{slackness_1} and \eqref{slackness_2}, we have
\begin{equation}
\frac{{\sqrt {{p_k}} }}{\gamma_k}  = \frac{{\sum\nolimits_{i = 1}^N {\sqrt {{p_i}} {\gamma _i}} }}{{\mu}{\eta }-\lambda_k},
\end{equation}
where ${\nu _k}=\theta_k=0$. If $p_j<E_j$ ($C_j>0$), we know $\theta_j=0$ from \eqref{slackness_2}, leading to the following results.
\begin{itemize}
\item If $p_j=p_{max}$, we know $\lambda_j=0$ from \eqref{slackness_1} and have
\begin{equation}
\frac{{\sqrt {{p_j}} }}{\gamma_j}  = \frac{{\sum\nolimits_{i = 1}^N {\sqrt {{p_i}} {\gamma _i}} }}{{\mu}{\eta } +\nu_j} < \frac{{\sqrt {{p_k}} }}{\gamma_k}.
\end{equation}
As ${\gamma_j}<{\gamma_k}$, $p_j$ should be less than $p_k$ ($p_j<p_k=0$) to ensure the above relationship, which conflicts the original assumption $p_j=p_{max}$. In other words, this power allocation policy $p_j=p_{max}$ cannot be optimal.
\item If $0<p_j<p_{max}$, $\lambda_j=\nu_j=0$ from \eqref{slackness_1} and
\begin{equation}
\frac{{\sqrt {{p_j}} }}{\gamma_j}  = \frac{{\sum\nolimits_{i = 1}^N {\sqrt {{p_i}} {\gamma _i}} }}{{\mu}{\eta}} < \frac{{\sqrt {{p_k}} }}{\gamma_k}.
\end{equation}
As ${\gamma_j}<{\gamma_k}$, $p_j$ should be less than $p_k$ ($p_j<p_k=0$) to ensure the above relationship, which conflicts the original assumption $0<p_j<p_{max}$. Hence, this policy $0<p_j<p_{max}$ cannot be optimal.
\end{itemize}
Hence, the optimal policy for the $j$-th grid-charging RAU is $p_j^*=0$. For the grid-discharging RAUs, we can draw the same conclusion. These indicate that $p_j=0$ as long as ${\gamma_j}<{\gamma_k}$.
\end{proof}

When the updated channel gain is worse, the SWIPT performance achieved using the harvesting energy at the current RAU is lower than the one achieved by trading the energy with the other RAUs. This is similar to the traditional water-filling algorithm, where water finds its level when filled in a vessel with multiple openings until dripping the last drop of water, showing that the power is always allocated to the one with better channel gain to earn more worth.

\begin{lem}
For any $k$ and $j$, if the optimal power of the $k$-th grid-charging ($p_k^*<E_k$) RAU is greater than zero but less than $p_{max}$, the power for the $j$-th grid-charging RAU having worse SNR than the $k$-th RAU is determined as $p_j^* = \frac{{\gamma _j^2}}{{\gamma _k^2}}p_k^*$, and all the grid-charging RAUs have the same ratio between the power allocation and the updated channel gain, which can be shown as $\kappa_G=\frac{{\sqrt {{p_k}} }}{{{\gamma _k}}} = \frac{{\sqrt {{p_j}} }}{{{\gamma _j}}}=\frac{{\sum\limits_{i = 1}^N {\sqrt {{p_i}} {\gamma _i}} }}{{\mu \eta }}$. The grid-discharging RAUs have the same property, but at a lower ratio given by $\kappa_L=\frac{{\sqrt {{p_k}} }}{{{\gamma _k}}} = \frac{{\sqrt {{p_j}} }}{{{\gamma _j}}}= \frac{{\sum\limits_{i = 1}^N {\sqrt {{p_i}} {\gamma _i}} }}{{\mu /\eta }}={\eta^2}{\kappa_G}$.
\end{lem}

\begin{proof}
For $0<p_k^*<p_{max}$, we obtain $\lambda_k=\nu_k=0$. When the grid is charging, it is clear that $\theta_k=0$. From \eqref{property_1}, we then have
\begin{equation}
\frac{{\sqrt {{p_k}} }}{{{\gamma _k}}} = \frac{{\sum\limits_{i = 1}^N {\sqrt {{p_i}} {\gamma _i}} }}{{\mu \eta }}.
\end{equation}
When $p_j<E_j$, we get
\begin{equation}
\frac{{\sqrt {{p_j}} }}{{{\gamma _j}}} = \frac{{\sum\limits_{i = 1}^N {\sqrt {{p_i}} {\gamma _i}} }}{{\mu \eta }-\lambda_j+\nu_j}.
\end{equation}
If $p_j \neq 0 $, then $\lambda_j=0$ and
\begin{equation}
      \frac{{\sqrt {{p_j}} }}{{{\gamma _j}}} = \frac{{\sum\limits_{i = 1}^N {\sqrt {{p_i}} {\gamma _i}} }}{{\mu \eta }+\nu_j} \le \frac{{\sqrt {{p_k}} }}{{{\gamma _k}}}
\end{equation}
for $\gamma_j<\gamma_k$, so $p_j<p_k<p_{max}$ and $\nu_j=0$. Then we have
\begin{equation}
      \kappa_G=\frac{{\sqrt {{p_k}} }}{{{\gamma _k}}} = \frac{{\sqrt {{p_j}} }}{{{\gamma _j}}}=\frac{{\sum\limits_{i = 1}^N {\sqrt {{p_i}} {\gamma _i}} }}{{\mu \eta }}.
\end{equation}
From ${{\sqrt {{p_k}} }}/{{{\gamma _k}}} = {{\sqrt {{p_j}} }}/{{{\gamma _j}}}$, we obtain $p_j^* = \frac{{\gamma _j^2}}{{\gamma _k^2}}p_k^*$.
Similarly for the case $p_k>E_k$, $p_j>E_k$, we can obtain
\begin{equation}
\kappa_L=\frac{{\sqrt {{p_k}} }}{{{\gamma _k}}} = \frac{{\sum\limits_{i = 1}^N {\sqrt {{p_i}} {\gamma _i}} }}{\left(\frac{\mu}{\eta}\right)}.
\end{equation}
For $\eta<1$, thus $\kappa_L={\eta^2}{\kappa_G} <\kappa_G$.
\end{proof}

Due to Jensen's inequality, we also have
\begin{equation}\label{Jensen}
  {\left( {\sum\limits_{k = 1}^N {\sqrt {{p_k}} {\gamma _k}{\ell _k}} } \right)^2} \leq
   \sum\limits_{k = 1}^N {{\ell _k}{{\left( {\sqrt {{p_k}} {\gamma _k}} \right)}^2}},
\end{equation}
where $\sum\nolimits_{k = 1}^N {\ell _k}=1$ and equality holds if and only ${\sqrt {{p_1}} {\gamma _1}}={\sqrt {{p_2}} {\gamma _2}}=\cdots={\sqrt {{p_N}} {\gamma _N}}$. In other words, the optimal energy strategy without constraints is equal to ${\sqrt {{p_k}} {\gamma _k}}$ power allocation. Note that this policy is modified by the constraints and the electric power transmission efficiency, which leads to double thresholds $\kappa_G$ and $\kappa_L$.

We also note that when $0 < p_k < p_{max}$, from the equality in \eqref{property_1},  we have ${{\sqrt {{p_k}} }}/{{{\gamma _k}}}\le {\kappa_G}$ since $\theta_k \ge 0$. Similarly, from the equality in \eqref{property_2}, we have ${{\sqrt {{p_k}} }}/{{{\gamma _k}}}\ge {\kappa_L}$ since $\xi_k \ge 0$. Therefore, for $0 < p_k < p_{max}$, we have
\begin{align}\label{relation}
    {\kappa_L} \le p_k \le {\kappa_G}.
\end{align}
We observe from {\em Lemma 2} that we have either $C_i>0$ and $D_i=0$, or $C_i=0$ and $D_i>0$, or $C_i=0$ and $D_i=0$. When $C_i=D_i=0$, from \eqref{Power}, we have $p_i=E_i$ which must satisfy \eqref{relation}. These show that there is a double-threshold policy on $p_i$. Specifically, when the grid is being charged, the transmit power equals the charging threshold $\kappa_G$; and when the grid is being discharged, the transmit power equals the discharging threshold $\kappa_L$. If the grid is neither being charged nor discharged, i.e., passive grid, $p_i=E_i$, or the transmitter uses up all the harvested energy at the current RAU.

\begin{theo}\label{thresholds}
The power policy solving \eqref{objective2} has the following double-threshold structure:
\begin{itemize}
  \item If the grid is charging, i.e., $E_k>\gamma_k^2\kappa_G^2$, $p_k=\gamma_k^2\kappa_G^2$.
  \item If the grid is discharging, i.e., $E_k<\gamma_k^2\kappa_L^2$, $p_k=\gamma_k^2\kappa_L^2$.
  \item If the grid is passive, i.e., $\gamma_k^2\kappa_L^2 \le E_k \le \gamma_k^2\kappa_G^2$, $p_k=E_k$.
\end{itemize}
\end{theo}

In summary, \textit{Theorem \ref{thresholds}} shows that the optimal policy for the case $p_k^*\left|_ {0 < p_k^* < p_{\max }}\right.$ can be calculated by
\begin{equation}\label{threshold_eq}
   p_{k}^*  = \left\{
      \begin{array}{ll} 
          \gamma_k^2(\kappa_G^*)^2, & \mbox{if }E_k > \gamma_k^2(\kappa_G^*)^2,\\
          \gamma_k^2(\kappa_L^*)^2, & \mbox{if }E_k < \gamma_k^2(\kappa_L^*)^2,\\
          E_k,  & \mbox{otherwise}.
      \end{array} \right.
\end{equation}
Based on \emph{Lemmas 4--6} and \emph{Theorem 2}, we can characterize the optimal policy. To find the entire policy, let us first consider the case $\left( {p_k^ * ,\lambda _k^ * ,\nu _k^ *, \mu^*} \right) =\left( {p_k^ * ,0,0, \mu^*} \right)\left| {0 < p_k^ *  < {p_{\max }}} \right.$. To obtain the optimal power allocation of this case, we need to find the thresholds ${\kappa_G^*}$ and ${\kappa_L^*}$ for $\forall k \in {1,\ldots,N}$ in light of \emph{Theorem 2}. As the relationship between ${\kappa_G}$ and ${\kappa_L}$,  $\kappa_L=\eta^2{\kappa_G}$ based on \emph{Lemma 5}, we only need to find the threshold ${\kappa_G}$, which will be realized by a one-dimensional linear search. We can continuously increase the value of $\kappa_G$ until the sum trade states below reach to zero:
\begin{equation}
  \sum\limits_{k = 1}^N {{S_k}}  = \sum\limits_{k = 1}^N {\left( {\eta {{\left[ {{E_k} - \gamma _k^2\kappa _G^2} \right]}^ + } - \frac{1}{\eta}{\left[ {\gamma _k^2\kappa _L^2 - {E_k}} \right]}^ +} \right) = 0}.
\end{equation}
Here, it should be noted that the power cannot exceed $p_{max}$. Therefore, we should check that if there is any power greater than and equal to $p_{max}$ based on the value of $\kappa_G$, then set them to $p_{max}$ and remove them from the searching set. From \emph{Lemma 5}, we can quickly find out all the RAUs which will be allocated maximum power, and of course these RAUs will also be removed from the searching set. Since the power for the rest of RAUs (denoted as set $\{{Re}\}$) need to be redecided, we then determine the solutions in a similar way for the rest of RAUs with the new sum constraint of the trade states
\begin{equation}\label{newsum}
  \sum\nolimits_{k = 1\atop k \ne j}^N {{S_k}}  + \sum\nolimits_{{j \in N\left| {{p_j} = {p_{\max }}} \right.} } {{S_j} = 0},
\end{equation}
where $S_j=\left( {\eta {{\left[ {{E_j} - {p_{\max }}} \right]}^ + } - {{\left[ {{p_{\max }} - {E_j}} \right]}^ + }/\eta } \right)$. We repeat this  until there is no power greater than and equal to $p_{max}$ for the rest of RAUs. Remarkably, in this process, we can find out which RAUs will be allocated full power (referred to as \emph{Algorithm 1}). Next, it remains to examine whether there is zero-power RAU or not (\emph{Algorithm 2}). From \emph{Lemma 4}, we know that if the last RAU is not allocated zero power, then there is no zero-power RAU. Therefore, we only check the last RAU for each loop by comparing the objective function of having or not having zero power RAU.  After this examination, the optimal policy is found. Combining with \emph{Lemma 3}, we have the overall algorithm summarized in \emph{Algorithm 3}.

\begin{algorithm}[!t]
\caption{Finding the Full-Power RAUs}
\label{alg1}
\begin{algorithmic}[1]
\firstcon ~~\\
$S_{Max}$, $\{Re\}$
\ENSURE
\STATE For all $k \in \{Re\}$;
\IF {$p_{Re(k)} \ge p_{max}$}
\STATE ${\{Re\}}\leftarrow {\{Re\}}-{\{Re(k)\}}$;
\STATE Set {$p_{Re(k)} = p_{max}$};
\STATE $S_{Max}\leftarrow S_{Max} +\eta{C_{Re(k)}} -{D_{Re(k)}}/{\eta}$;
\IF {$p_{Re(k)} > E_{Re(k)}$}
\STATE For all $j<k$;
\IF {$E_{Re(j)} > p_{max}$}
\STATE ${\{Re\}}\leftarrow {\{Re\}}-{\{Re(j)\}}$;
\STATE Set {$p_{Re(j)} = p_{max}$};
\STATE $S_{Max}\leftarrow S_{Max} +\eta{C_{Re(j)}} -{D_{Re(j)}}/{\eta}$;
\ENDIF\\
\ELSIF  {$p_{Re(k)} < E_{Re(k)}$}
\IF {$E_{Re(j)} \ne p_{max}$}
\STATE  ${\{Re\}}\leftarrow {\{Re\}}-{\{Re(j)\}}$;
\STATE Set {$p_{Re(j)} = p_{max}$};
\STATE $S_{Max}\leftarrow S_{Max} +\eta{C_{Re(j)}} -{D_{Re(j)}}/{\eta}$;
\ENDIF\\
\ENDIF\\
\ENDIF\\
\lastcon ~~\\
$S_{Max}$, $\{Re\}$
\end{algorithmic}
\end{algorithm}

\begin{algorithm}[!t]
\caption{Finding the Zero-Power RAUs}
\label{alg2}
\begin{algorithmic}[1]
\firstcon ~~\\
$S_{Max}$, $\{Re\}$, $p_{Re}$, $N_{Re}$
\REQUIRE ~~\\
Set $Target=0$, $\bar{p}_{Re}=p_{Re}$, $\bar{S}_{Max}=S_{Max}$;
\ENSURE
\STATE Set $k = N_{Re}$
\STATE Set $\bar{p}_{Re(k)}=0$;
\STATE $\bar{S}_{Max}\leftarrow \bar{S}_{Max} +\eta{E_{Re(k)}}$;
\STATE Begin from RAU $j=1, j\in {\{Re\}-\{Re(k)\}}$, find the thresholds $\kappa_G$ and $\kappa_L$ that enable $\sum_{k \in {\{Re\}}}S_k+\bar{S}_{Max}=0$ with the \textit{Theorem \ref{thresholds} } by a one-dimensional search in \cite{Yuan} or by dynamic programming \cite{DP};
\WHILE{$j \le N_{Re}$}
\STATE Compute $\bar{p}_{Re(j)}$ using equation (\ref{threshold_eq});
\STATE $j\leftarrow j+1$;
\ENDWHILE\\
\STATE $Target$ is equal to the calculation value of \eqref{objective2_1} with respect to the new $\bar{p}_{Re(j)}$;
\lastcon ~~\\
$Target$, $\bar{p}_{Re}$, $\bar{S}_{Max}$
\end{algorithmic}
\end{algorithm}

\begin{algorithm}[!t]
\caption{Optimal Double-Threshold DAS with SWIPT}
\label{alg3}
\begin{algorithmic}[1]
\REQUIRE ~~\\
Set $k=1$, $S_{Max}=0$, $\{Re\}=\{1,2,\ldots,N\}$, $N_{Re}=length(Re)$, $Target_{opt}=0$;

\ENSURE
\IF {$\sum\limits_{i = 1}^N {{E_i}}  \ge N{p_{\max }} + \left( {1 - {\eta ^2}} \right)\left[ {\sum\limits_{i \in G} {\left( {{E_i} - {p_{\max }}} \right)} } \right]$}
\WHILE{$k \le N$}
\STATE $p_k^* = p_{max}$;
\STATE $k\leftarrow k-1$;
\ENDWHILE\\
\ELSE
\STATE Begin from RAU $k=1, k\in {\{Re\}}$, find the thresholds $\kappa_G$ and $\kappa_L$ that enable $\sum_{k \in {\{Re\}}}S_k+S_{Max}=0$ with the \textit{Theorem \ref{thresholds} } by a one-dimensional search in \cite{Yuan} or by dynamic programming \cite{DP};
\WHILE{$k \le N_{Re}$}
\STATE Compute $p_{Re(k)}$ using equation (\ref{threshold_eq});
\STATE $k\leftarrow k+1$;
\ENDWHILE\\
\STATE do \emph{Algorithm 1};
\IF {$N_{Re}=length(Re)$}
\STATE $Target_{opt}$ is equal to the calculation value of \eqref{objective2_1} with respect to $p_{Re(k)}$;
\STATE do \emph{Algorithm 2};
\IF {$Target>Target_{opt}$}
\STATE $Target_{opt}=Target$, $p_{Re}=\bar{p}_{Re}$, $S_{Max}=\bar{S}_{Max}$
\STATE Go to Step 15;
\ELSE
\STATE $p^*=p$;
\STATE Step out of the iteration;
\ENDIF
\ELSE
\STATE Set $N_{Re}=length(Re)$
\STATE Go to Step 7;
\ENDIF\\
\ENDIF\\
\lastcon ~~\\
\STATE $p_k^*, \forall k \in \{1, 2, \ldots, N\}$;
\end{algorithmic}
\end{algorithm}

It is emphasized that our proposed algorithm is applicable to DAS regardless of the number of RAUs $N$ and the number of antennas per RAU $M$. For illustration, consider a DAS system that has $16$ RAUs with $4$ antenna each, for which the electric power transmission efficiency is $\eta=0.8$ and the per-RAU power constraint is $p_{max}=5$. The harvested power of each RAU is given by $E =[6,2,6,4,1,1,4,5,1,1,4,8,1,8,1,4]$ and the corresponding updated channel gain is in descending order. We employed \emph{Algorithm 3} to determine the optimal power allocation with results as depicted in Fig.~\ref{example_for_16}. In this chart, the optimally allocated powers are indicated by stems while the corresponding threshold levels are shown by bars. Grid-charging, discharging and passive RAUs are denoted by the blue, green and yellow bars, respectively. For the case $\left( {p_k^ * ,\lambda _k^ * ,\nu _k^ *, \mu^*} \right) =\left( {p_k^ * ,0,0, \mu^*} \right)\left| {0 < p_k^ *  < {p_{\max }}} \right.$ (i.e., $(5\text{-}14)$-th RAU), we note that the charging RAUs (i.e., $(11, 12, 14)$-th RAU) have the same threshold $\kappa_G$ power roughly equal to $22.7445$. RAUs $(5,6,9,10)$ on the other hand undergo the discharging process, with the same lower threshold $\kappa_L$ roughly equal to $14.5565$. The power allocation of RAUs $(7,8,13)$ is shown between charing and discharging one, which is called grid-passive. We also observe that the full-power RAUs (i.e., $(1\text{-}4)$-th RAU) and zero-power RAUs (i.e., $(15,16)$-th RAU ) accord with the obtained results in \emph{Lemmas 4 and 5}.

\begin{figure}[!t]
\centering
\includegraphics[scale=0.45]{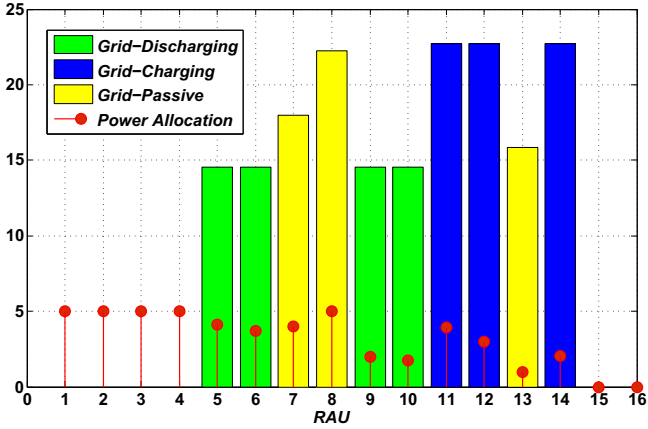}
\caption{Example illustrating the optimal policy with thresholds $\kappa_G$ and $\kappa_L$, with the number of RAUs $N=16$, the number of antennas per RAU $M=4$, and the electric power transmission efficiency $\eta=0.8$.}\label{example_for_16}
\end{figure}

Based on the above, now we can obtain the maximum WIT performance with the value of the splitting power ratio which is calculated by the minimum WET constraint as
\begin{equation}
  \rho  = 1 - \frac{{{Q^{\min }}}}{{\xi \left({{\left( {\sum\limits_{i = 1}^N {\sqrt {p_i^*}{\gamma_i}}} \right)^2}+\sigma^2}\right)}}.
\end{equation}
By substituting this into \eqref{ID}, the whole problem is solved.

\section{Numerical Results}
In this section, we provide the simulation results to evaluate the performance of the proposed algorithm under different settings. We considered the optimal value of \eqref{objective2_1} as the SWIPT performance. We assume Rayleigh fading channels and in the simulations, we provide the results by averaging over $1000$ independent randomly generated channel realizations and $1000$ randomly corresponding energy harvested realizations for each point with fixed $\alpha=2$. We also assume that all the RAUs are randomly and uniformly deployed in the distance from the user of range $(10,50)$ meters. Each RAU has a random energy arrival uniformly distributed over $[1,8]$, denoted as ${\mathcal U}(1,8)$.

We start by examining the SWIPT performance as a function of $N$ for DAS with a different number of antennas per RAU in Fig.~\ref{example_for_N_M1234}. The per-RAU power constraint is set as $p_{max}=5$ and the electric power transmission efficiency is set as $\eta=0.8$. It is shown that we can achieve higher SWIPT as $N$ increases and the performance gap becomes larger with the number of antennas. This is due to the fact that when the number of RAUs is small, the average distance between the RAU and the user will be large, and the large-scale fading will correspondingly be large, thereby inferior performance from the start. As the the number of RAUs increases, the average distance becomes shorter and the performance will increase significantly.

\begin{figure}[!t]
\centering
\includegraphics[scale=0.52]{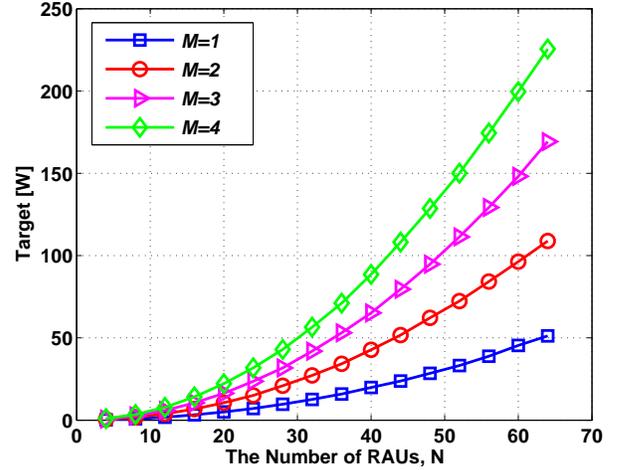}
\caption{The WET performance versus the number of RAUs and the number of antennas per RAU with $E_k \in {\mathcal U}(1,8), \forall k$.}\label{example_for_N_M1234}
\end{figure}

Fig.~\ref{example_for_N_pmax3456} illustrates the SWIPT performance with different per-RAU power constraints. The number of antennas per RAU is set as $M=4$ and the electric power transmission efficiency is set as $\eta=0.8$. The SWIPT performance of our algorithm gradually improves and the performance increases slowly as $p_{max}$ increases. This is because all the power allocated will be smaller than $p_{max}$ with high probability, and as a result, $p_k^*$'s will not change any more. This is especially the case when the required SINR is lower.

\begin{figure}[!t]
\centering
\includegraphics[scale=0.52]{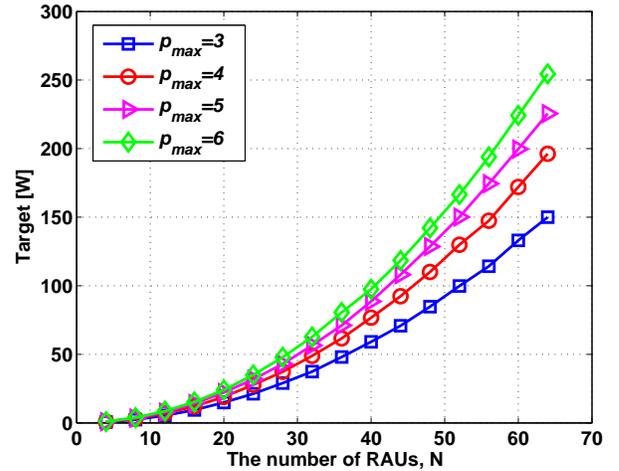}
\caption{The WET performance versus the number of RAUs and the per-RAU power constraint with $E_k \in {\mathcal U}(1,8), \forall k$.}\label{example_for_N_pmax3456}
\end{figure}

Fig.~\ref{example_for_N_ee57910} shows the SWIPT performance for different electric power transmission efficiency. The number of antennas per RAU is set as $M=4$ while the per-RAU power constraint is set as $p_{max}=5$. We see that through our optimal proposed policy, even if the electric power transmission efficiency is small, the SWIPT will not decrease much. We also observe that the SWIPT performance of $\eta=0.9$ will be very close to the perfect case $\eta=1$. This is because that we choose to share the energy to achieve more benefit with the modified Jensen's inequality policy. When the electric power transmission efficiency is small, the optimal policy will use the energy harvested directly instead of sharing it to avoid energy loss.

\begin{figure}[!t]
\centering
\includegraphics[scale=0.52]{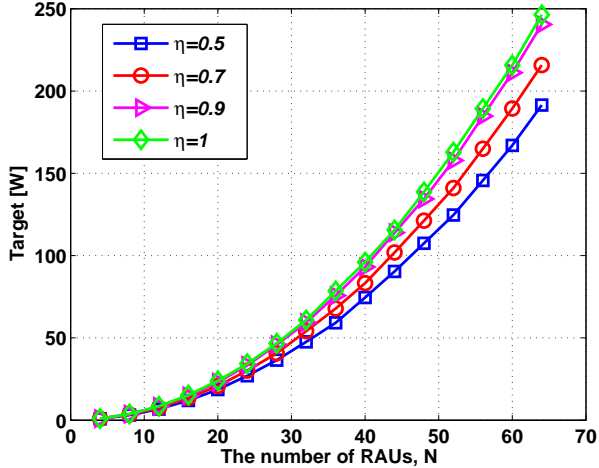}
\caption{The WET performance versus the number of RAUs and the electric power transmission efficiency with $E_k \in {\mathcal U}(1,8), \forall k$. }\label{example_for_N_ee57910}
\end{figure}

To provide some comparisons to help understand the performance of our proposed scheme, we consider several suboptimal transmission policies as banchmearks in Fig.~\ref{example_for_N_ee57910}:
\begin{itemize}
\item Greedy---The EH device directly fuels the RAU if energy is sufficient for maximum power transmission and the excessive energy charges the grid; otherwise the RAU uses up all the current harvested energy. Then the grid discharges the charging energy to the RAU which transmits with power lower than $p_{max}$ but has the best update channel gain among the non-maximum power RAUs. The grid continues to discharge any charging energy left.
\item Water-filling---An adaptive water-filling policy is found by adapting the power to the undate channel gain $\gamma$. The power allocated to each RAU is calculated by
\begin{equation}
{p_k} = \min \left( {{p_{\max }},{{\left[ {\varsigma  - \frac{1}{{{\gamma _k}}}} \right]}^ + }} \right),
\end{equation}
where the cutoff water level $\varsigma$ is calculated as the solution of the following equation:
\begin{equation}
\sum\limits_{k = 1}^N {{S_k}}  = \sum\limits_{k = 1}^N \left(\eta \left[E_k - p_k\right]^ + - \frac{1}{\eta}\left[p_k - E_k\right]^+\right)=0.
\end{equation}
\end{itemize}

Results in Fig.~\ref{comparison} indicate that the proposed optimal policy outperforms the other two schemes, regardless of the number of RAUs and the electric power transmission efficiency. Note that the water-filling algorithm performs worse if the electric power transmission efficiency is low $\eta=0.8$ because it does not take $\eta$ into account which leads to large energy loss due to sharing. Also, as expected, the greedy policy performs slightly worse than the optimal one. That is, in the low $\eta$ case, there will be fewer energy sharing among the RAUs to avoid energy loss. The greedy policy itself mainly focuses on using up the current harvesting energy first. Thus, these two strategies have similar management and performance.

\begin{figure}[!t]
\centering
\includegraphics[scale=0.52]{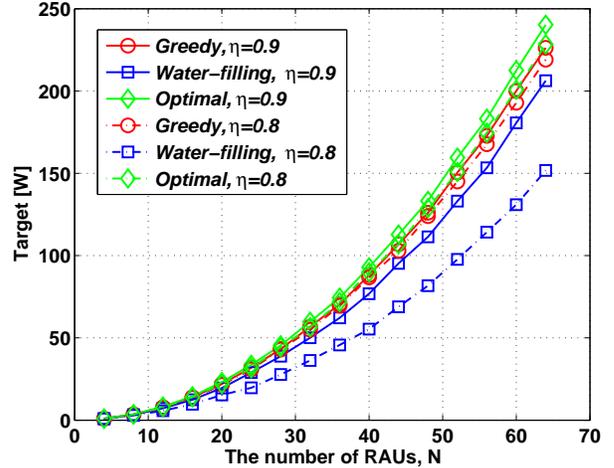}
\caption{Comparison of the SWIPT performance with different policies versus the number of RAUs, with $p_{max}=5$, $M=4$, $E_k \in {\mathcal U}(1,8), \forall k$.}\label{comparison}
\end{figure}

Finally, we provide the WIT-WET region versus the PS ratio $\rho$. It is assumed that $\zeta=0.5$, $\sigma^2=1$ , and $\tau^2=1$. The PS ratio increases from left to right in the figure with the range $[0,1]$. We observe that with $\rho$ increasing, the WIT (rate) improves with the WET performance drops. It is shown that when the number of the RAUs is large, the increasing rate of the WIT performance is faster than WET's.

\begin{figure}[!t]
\centering
\includegraphics[scale=0.52]{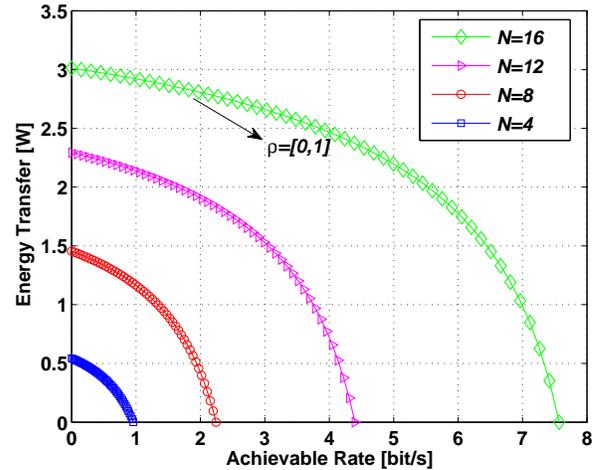}
\caption{The WIT-WET region versus the PS ratio, with $p_{max}=5$, $\eta=0.8$, $M=4$, $E_k \in {\mathcal U}(1,8), \forall k$. }\label{Region}
\end{figure}

\section{Conclusion}
This paper investigated the optimal energy cooperation policy for SWIPT DAS. Optimization was done in the framework of WIT maximization, subject to minimum WET constraint as well as energy causality and green smart constraints. From the WIT and WET formulation, we showed that the optimization can be solved by maximizing ${| {\sum\nolimits_{i = 1}^N {\sqrt {{p_i}} } \mathbf{g}_{i}^H{\mathbf{w}_{i}}} |^2}$, first with the PS ratio $\rho$ fixed and then adjusting it in accordance with the WET constraint. It was also revealed that the former problem can be solved by dividing the green smart constraint into the smart grid profitable, and the smart grid neutral cases. A full-power transmission strategy was derived in the former case. As to the latter one, we demonstrated that the optimal policy takes one of the three power allocation forms: maximum power allocation, zero power allocation, and a mix between the two. Each form has its property, especially for the last one, which has a double-threshold structure. Based on this, we proposed a double-threshold strategy to solve the entire problem and provided an algorithm to efficiently find the solution. Numerical results were presented to validate the theoretical analysis and to demonstrate the superior performance of the optimal proposed policy over other two schemes in the literature.

%

\end{document}